\documentclass[]{llncs}

\usepackage[]{textcomp}
\usepackage[]{microtype}
\usepackage[utf8]{inputenc}
\usepackage[T1]{fontenc}
\usepackage[]{stmaryrd}
\usepackage[]{amssymb}
\usepackage[]{amsmath}

\usepackage{hyperref}

\title{Abstract interpretation as anti-refinement}

\author{Arnaud Spiwack}

\institute{Inria Paris-Rocquencourt\\\textsc{Ens}, 45 rue d'Ulm, 75230 Paris Cedex 05, France\\
\email{arnaud@spiwack.net}}

\begin{document}
  \maketitle
  \begin{abstract}
    This article shows a correspondence between abstract interpretation of imperative programs and the refinement calculus: in the refinement calculus, an abstract interpretation of a program is a specification which is a function.
    \par
    This correspondence can be used to guide the design of mechanically verified static analyses, keeping the correctness proof well separated from the heuristic parts of the algorithms.
  \end{abstract}
  \section{Introduction}
  A mathematical way to describe a static analysis is to see it as a program which tries to prove a theorem about programs. It may fail to do so, but if it succeeds, it effectively acts as a proof of the said theorem. The proof, however, is essentially impossible to check by a human.
  \par
  To increase the level of trust in a static analysis tool, the tool can be mechanically verified, for instance in Coq~\cite{coq}, thus ensuring that the produced proof is always correct. In the design of a static analysis tool, some parts are crucial for correctness, while other are heuristic. For instance, a static analysis can choose to lose precision to gain performance. Hence, from the point of view of he who wants to ensure the correctness, a static analysis can be seen as an interplay between a correctness enforcer and an heuristic-providing oracle. The question addressed in this article is how to formalise this interplay.
  \par
  To that end, we use the refinement calculus~\cite{Back1998,VonWright1994}. The refinement calculus is a well-established method for proving program properties. It comes with a natural notion of interaction, generally used to model the interaction between the implementer of a unit of code and its user. In the context of this article, the correctness enforcer plays the role of the implementer while the oracle is the user.
  \par
  Specifically, this article shows the connection between static analysis by abstract interpretation~\cite{Cousot1992} and the refinement calculus. Namely, it shows that an abstract domain constructs a \emph{specification} of the analysed program, which happens to be given by a function. This correspondence is instrumental in the design of Cosa~\cite{cosa}, a Coq formalisation of a shape analysis.
  \par
  The two subjects have some notation overlap, hence some unconventional notations will be used. The author apologises, but hopes that practitioners of both subjects will not find the notations too surprising or confusing.
  \section{Predicate transformers}
  \par
  Edsger Dijkstra introduced the idea of using predicate transformers as semantics of imperative programs~\cite{Dijkstra1978}. The idea is to associate to each program $p$ a function $\mathsf{wlp}{\left( p\right) }$, its \emph{weakest liberal precondition} operator, such that for a property $P$ of program states, $\mathsf{wlp}{\left( p\right) }{\left( P\right) }$ is the weakest condition on the initial state, such that after running $p$, if $p$ terminates, then $P$ holds.
  \par
  Weakest liberal precondition accounts for partial correctness. Alternatively, one could use the weakest precondition operator (which additionally imposes that $p$ terminates) to account for total correctness. Termination is not our purpose here, and we will identify programs with their weakest liberal precondition operator.
  \par
  Predicate transformer semantics is the starting point of refinement calculus~\cite{Back1998}, and is also commonly used in abstract interpretation -- see~\cite{Cousot1997a} for a discussion of weakest liberal precondition in relation to abstract interpretation.
  \par
  \subsection{Basic definitions}
  \par
  We will call \emph{predicate transformers} monotone functions in $\mathcal{P}{\left( A\right) }\rightarrow \mathcal{P}{\left( B\right) }$ for some sets $A$ and $B$, and write $\mathcal{P}{\left( A\right) }{\rightarrow }^{+}\mathcal{P}{\left( B\right) }$ for the set of predicate transformers. The set $\mathcal{P}{\left( A\right) }\rightarrow \mathcal{P}{\left( B\right) }$ inherits the complete lattice structure of $\mathcal{P}{\left( B\right) }$ and $\mathcal{P}{\left( A\right) }{\rightarrow }^{+}\mathcal{P}{\left( B\right) }$, equipped with the lattice operations of $\mathcal{P}{\left( A\right) }\rightarrow \mathcal{P}{\left( B\right) }$, is also a complete lattice. We write $a\sqsubseteq b\iff \forall {X}^{\in \mathcal{P}{\left( A\right) }}.\,\,a{\left( X\right) }\subseteq b{\left( X\right) }$ for the inherited order.
  \par
  We shall call the following operations of predicate transformers \emph{regular operations}. They have a direct interpretation as program constructs. Programs will be interpreted as homogeneous predicate transformers $\mathcal{P}{\left( A\right) }{\rightarrow }^{+}\mathcal{P}{\left( A\right) }$, however the regular operations also work with general predicate transformers $\mathcal{P}{\left( A\right) }{\rightarrow }^{+}\mathcal{P}{\left( B\right) }$.
  \begin{description}
    \item[Sequence] $\left( a;b\right) {\left( X\right) }=a{\left( b{\left( X\right) }\right) }$
    \par
                  Reads as ``do $a$ then do $b$''. The definition of sequence emphasises the fact that the weakest liberal precondition semantics is a \emph{backward} semantics. Sequence is associative, and monotone:
                  \begin{itemize}
      \item $\left( a;b\right) ;c=a;\left( b;c\right) $
      \item $a\sqsubseteq a'\land b\sqsubseteq b'\implies a;b\sqsubseteq a';b'$
    \end{itemize}
    \item[Skip] $1{\left( X\right) }=X$
    \par
              Does not do anything. Skip is neutral for sequence:
              \begin{itemize}
      \item $1;a=a=a;1$
    \end{itemize}
    \item[Choice] $\left( a+b\right) {\left( X\right) }=a{\left( X\right) }\cap b{\left( X\right) }$
    \par
                Non-deterministic choice. Choice is associative, commutative and monotone. Moreover sequence distributes on the right over choice:
                \begin{itemize}
      \item $\left( a+b\right) +c=a+\left( b+c\right) $
      \item $a+b=b+a$
      \item $\left( a+b\right) ;c=a;c+b;c$
      \item $p\sqsubseteq \left( a+b\right) ;q\iff p\sqsubseteq a;q\land p\sqsubseteq b;q$
    \end{itemize}
    \item[Hang] $0{\left( X\right) }=\top $
    \par
               Hang loops indefinitely. It is neutral for choice, sequence distributes on the right over it, and it is the largest predicate transformer:
               \begin{itemize}
      \item $0+a=a=a+0$
      \item $0;a=0$
      \item $a\sqsubseteq 0$
    \end{itemize}
    \item[Iteration] ${a}^{*}$, for $a\in \mathcal{P}{\left( A\right) }{\rightarrow }^{+}\mathcal{P}{\left( A\right) }$, is the largest fixed point of the (monotone) function which maps $p$ to $1+a;p$. It runs $a$ in sequence a non-deterministic number of times (including none, and infinitely many). It has the following properties~\cite[Chapter 21]{Back1998}:
               \begin{itemize}
      \item ${a}^{*};q=q+a;{a}^{*};q$
      \item $p\sqsubseteq q+a;p\implies p\sqsubseteq {a}^{*};q$
    \end{itemize}
  \end{description}
  \par
  It should be noted that despite the name ``regular operations'', predicate transformers do not form a Kleene algebra under these operations. Indeed the left distributivity laws are missing: $a;\left( b+c\right) =a;b+a;c$ and $a;0=0$ do not hold in general.
  \par
  \subsection{Programs}\label{latex_lib_label_1}
  \par
  In this setting, a programming language consists in a set $\mathcal{S}$ of states together with a set $\mathcal{I}\subseteq \mathcal{P}{\left( \mathcal{S}\right) }{\rightarrow }^{+}\mathcal{P}{\left( \mathcal{S}\right) }$ of \emph{basic instructions}. A program in the language $\left( \mathcal{S},\mathcal{I}\right) $ is an element of the subset of $\mathcal{P}{\left( \mathcal{S}\right) }{\rightarrow }^{+}\mathcal{P}{\left( \mathcal{S}\right) }$ generated by $\mathcal{I}$ and the regular operations.
  \par
  The use of non-deterministic choice and iterations make the programs non-deterministic. This is a natural setting for both program refinement and abstract interpretation. However, a typical programming language will feature a set of tests $\mathcal{B}$ such that for all $b\in B$, there is $\llbracket b\rrbracket \in \mathcal{P}{\left( \mathcal{S}\right) }$, and $\mathsf{guard}{\left( b\right) }$ is an instruction, such that $s\in \mathsf{guard}{\left( b\right) }{\left( X\right) }\iff s\in \llbracket b\rrbracket \implies s\in X$.
  \par
  With this assumption, the usual deterministic programming constructs can be recovered: $\mathsf{if}~b~\mathsf{then}~u~\mathsf{else}~v=\left( \mathsf{guard}{\left( b\right) };u\right) +\left( \mathsf{guard}{\left( \neg b\right) };v\right) $, and $\mathsf{while}~b~\mathsf{do}~u={\left( \mathsf{guard}{\left( b\right) };u\right) }^{*};\mathsf{guard}{\left( \neg b\right) }$.
  \par
  \par
  \begin{example}
    As an example, let us consider a language with a single memory cell containing an integer. In other words, $\mathcal{S}=\mathbb{Z}$. It has two tests, $\mathsf{pos}$ and $\mathsf{npos}$, whose semantics are given by:
    \begin{itemize}
      \item $\llbracket \mathsf{pos}\rrbracket \iff \left\{ n\,{\in }\,\mathbb{Z}\mid n>0\right\} $
      \item $\llbracket \mathsf{npos}\rrbracket \iff \left\{ n\,{\in }\,\mathbb{Z}\mid n\leqslant 0\right\} $
    \end{itemize}
    and a operation $\mathsf{dec}$, which decrements the integer held in the state. Its semantics is given by:
    \begin{itemize}
      \item $\mathsf{dec}{\left( X\right) }=\left\{ n\,{\in }\,\mathbb{Z}\mid n-1\in X\right\} $
    \end{itemize}
    \par
    This language expresses, for example, the simple program whose effect is to decrease the integer held in the state until it is non-positive. We shall call this program $d$:
    \begin{itemize}
      \item $d=\mathsf{while}~\mathsf{pos}~\mathsf{do}~\mathsf{dec}={\left( \mathsf{guard}{\left( \mathsf{pos}\right) };\mathsf{dec}\right) }^{*};\mathsf{guard}{\left( \mathsf{npos}\right) }$
    \end{itemize}
  \end{example}
  \par
  \par
  \subsection{Relations}
  \par
  A relation is usually seen as a subset of $A\times B$, however, it will be more convenient to see them, equivalently, as functions of $A\rightarrow \mathcal{P}{\left( B\right) }$.
  \par
  Given a relation $r\in A\rightarrow \mathcal{P}{\left( B\right) }$, we can extend it to a predicate transformer in two ways:
  \begin{itemize}
    \item $\left\langle r\right\rangle \in \mathcal{P}{\left( A\right) }{\rightarrow }^{+}\mathcal{P}{\left( B\right) }$ defined by $\left\langle r\right\rangle {\left( X\right) }=\bigcup _{\mbox{\scriptsize{$x$${\in}$$X$}}}r{\left( x\right) }$
    \item $\left[ r\right] \in \mathcal{P}{\left( B\right) }{\rightarrow }^{+}\mathcal{P}{\left( A\right) }$ defined by $\left[ r\right] {\left( Y\right) }=\left\{ x\,{\in }\,A\mid r{\left( x\right) }\subseteq Y\right\} $
  \end{itemize}
  The predicate transformers $\left\langle r\right\rangle $ and $\left[ r\right] $ form a Galois connection \emph{i.e.}:
  \begin{itemize}
    \item $\forall {X}^{\in \mathcal{P}{\left( A\right) }},{Y}^{\in \mathcal{P}{\left( B\right) }}.\,\,\left\langle r\right\rangle {\left( X\right) }\subseteq Y\iff X\subseteq \left[ r\right] {\left( Y\right) }$
  \end{itemize}
  or equivalently:
  \begin{itemize}
    \item $\forall {X}^{\in \mathcal{P}{\left( A\right) }}.\,\,X\subseteq \left[ r\right] {\left( \left\langle r\right\rangle {\left( X\right) }\right) }$
    \item $\forall {Y}^{\in \mathcal{P}{\left( B\right) }}.\,\,\left\langle r\right\rangle {\left( \left[ r\right] {\left( Y\right) }\right) }\subseteq Y$
  \end{itemize}
  In fact, every Galois connection between powersets is of that form. This is due to the general fact about complete lattices that a left adjoint -- like $\left\langle r\right\rangle $ -- is the same thing as a function which preserves joins. In the case of powersets, a function which preserves joins is characterised by its action on singletons, hence is of the form $\left\langle r\right\rangle $.
  \par
  Identifying a function $f$ to its graph, we hence have a Galois connection between $\left\langle f\right\rangle $ and $\left[ f\right] $. These are better known as the direct image and the inverse image of $f$, which we will write ${f}_{*}$ and ${f}^{-1}$ respectively. We shall make use of the following consequence of their being a Galois connection:
  \begin{itemize}
    \item $x\in {f}^{-1}{\left( X\right) }\iff f{\left( x\right) }\in X$
  \end{itemize}
  \par
  The properties of Galois connections can also be read directly in terms of the predicate transformer lattice:
  \begin{itemize}
    \item $\left\langle r\right\rangle ;p\sqsubseteq q\iff p\sqsubseteq \left[ r\right] ;q$
    \item $p;\left[ r\right] \sqsubseteq q\iff p\sqsubseteq q;\left\langle r\right\rangle $
    \item ${f}_{*};p\sqsubseteq q\iff p\sqsubseteq {f}^{-1};q$
    \item $p;{f}^{-1}\sqsubseteq q\iff p\sqsubseteq q;{f}_{*}$
  \end{itemize}
  \section{Abstract interpretation}\label{latex_lib_label_2}
  \par
  Abstract interpretation~\cite{Cousot1992} is a framework for static analysis in which the objects of study are called \emph{domains}. As general as the definitions in this section are, they fail to capture the full generality of abstract interpretation. However, they are sufficient for most purposes -- at least for imperative languages.
  \par
  Fixing a programming language $\left( \mathcal{S},\mathcal{I}\right) $, the powerset $\mathcal{P}{\left( \mathcal{S}\right) }$ is called the concrete domain and the interpretation of a program as a predicate transformer $\mathcal{P}{\left( A\right) }{\rightarrow }^{+}\mathcal{P}{\left( A\right) }$ is called the concrete semantics.
  \par
  A departure from common practice is that the concrete semantics, the weakest liberal precondition, is backward -- \emph{i.e.} a function from a set of final states to corresponding initial states -- whereas often the concrete semantics is chosen to be forward. This choice has been made to stay closer to the practice in refinement calculus. Having a backward concrete semantics does not, however, constrain the analysis to be backward too. In the rest of the paper we will mainly consider forward analysis. Moreover, forward semantics are usually constructed from a relational semantics, \emph{i.e.} they are of the form $\left\langle r\right\rangle $, in which case $\left[ r\right] $ will be our backward semantics.
  \par
   An abstract domain is a set ${\mathcal{S}}^{\sharp }$ together with a concretisation function $\gamma  : {\mathcal{S}}^{\sharp }\rightarrow \mathcal{P}{\left( \mathcal{S}\right) }$ and extra material to construct an \emph{abstract semantics} to each program. The abstract semantics of a program is a forward function ${p}^{\sharp } : {\mathcal{S}}^{\sharp }\rightarrow {\mathcal{S}}^{\sharp }$ which has the following correctness property:
  \begin{itemize}
    \item $\forall {{s}^{\sharp }}^{\in {\mathcal{S}}^{\sharp }}.\,\,\forall {S}^{\in \mathcal{P}{\left( \mathcal{S}\right) }}.\,\,S\subseteq \gamma {\left( {s}^{\sharp }\right) }\implies S\subseteq p{\left( \gamma {\left( {p}^{\sharp }{\left( {s}^{\sharp }\right) }\right) }\right) }$
  \end{itemize}
  Which can, equivalently be stated as:
  \begin{itemize}
    \item $\forall {{s}^{\sharp }}^{\in {\mathcal{S}}^{\sharp }}.\,\,\gamma {\left( {s}^{\sharp }\right) }\subseteq p{\left( \gamma {\left( {p}^{\sharp }{\left( {s}^{\sharp }\right) }\right) }\right) }$
  \end{itemize}
  This phrasing of the correctness property may look a bit contorted to the practitioner of abstract interpretation. It is the consequence of having a backward concrete semantics and a forward abstract semantics. When the concrete semantics is of the form $p=\left[ {p}_{0}\right] $, then this correctness property coincides with the more familiar one:
  \begin{itemize}
    \item $\forall {{s}^{\sharp }}^{\in {\mathcal{S}}^{\sharp }}.\,\,\left\langle {p}_{0}\right\rangle {\left( \gamma {\left( {s}^{\sharp }\right) }\right) }\subseteq \gamma {\left( {p}^{\sharp }{\left( {s}^{\sharp }\right) }\right) }$
  \end{itemize}
  \par
  Abstract domains are meant to be composed. For that reason, the abstract semantics ${p}^{\sharp }$ is computed out of more atomic functions, which are, in particular, stable by Cartesian product. Writing $s\leqslant s'\iff \gamma {\left( s\right) }\subseteq \gamma {\left( s'\right) }$ for the order induced on ${\mathcal{S}}^{\sharp }$ by the concretisation function, the abstract domain comes equipped with the following:
  \begin{description}
    \item[Join] An operator $\sqcup \in {\mathcal{S}}^{\sharp }\times {\mathcal{S}}^{\sharp }\rightarrow {\mathcal{S}}^{\sharp }$ such that:
    \begin{itemize}
      \item ${s}^{\sharp }\leqslant {s}^{\sharp }\sqcup {t}^{\sharp }$
      \item ${t}^{\sharp }\leqslant {s}^{\sharp }\sqcup {t}^{\sharp }$
    \end{itemize}
    \item[Post-fixed point] An operator $\mathsf{pfp}\in \left( {\mathcal{S}}^{\sharp }\rightarrow {\mathcal{S}}^{\sharp }\right) \rightarrow \left( {\mathcal{S}}^{\sharp }\rightarrow {\mathcal{S}}^{\sharp }\right) $ such that:
    \begin{itemize}
      \item $\forall {f}^{\in {\mathcal{S}}^{\sharp }\rightarrow {\mathcal{S}}^{\sharp }},{{s}^{\sharp }}^{\in {\mathcal{S}}^{\sharp }}.\,\,{s}^{\sharp }\leqslant \mathsf{pfp}{\left( f\right) }{\left( {s}^{\sharp }\right) }$
      \item $\forall {f}^{\in {\mathcal{S}}^{\sharp }\rightarrow {\mathcal{S}}^{\sharp }},{{s}^{\sharp }}^{\in {\mathcal{S}}^{\sharp }}.\,\,f{\left( \mathsf{pfp}{\left( f\right) }{\left( {s}^{\sharp }\right) }\right) }\leqslant \mathsf{pfp}{\left( f\right) }{\left( {s}^{\sharp }\right) }$
    \end{itemize}
    \par
    Typically, the post-fixed point operator is derived from a widening operator $\nabla \in {\mathcal{S}}^{\sharp }\times {\mathcal{S}}^{\sharp }\rightarrow {\mathcal{S}}^{\sharp }$, which has the following properties:
    \begin{itemize}
      \item ${s}^{\sharp }\leqslant {s}^{\sharp }\nabla {t}^{\sharp }$
      \item ${t}^{\sharp }\leqslant {s}^{\sharp }\nabla {t}^{\sharp }$
      \item For every increasing sequence ${\left( {x}_{n}\right) }_{n\in \mathbb{N}}$, the sequence ${\left( {y}_{n}\right) }_{n\in \mathbb{N}}$ defined by ${y}_{0}={x}_{0}$ and ${y}_{n+1}={y}_{n}\nabla {x}_{n+1}$ verifies $\exists {n}^{\in \mathbb{N}}.\,\,{y}_{n+1}\leqslant {y}_{n}$.
    \end{itemize}
    Then, taking, mutually recursively, ${x}_{0}={s}^{\sharp }$, ${x}_{n+1}=f{\left( {y}_{n}\right) }$, and ${y}_{n}$ such as above, we can then define $\mathsf{pfp}{\left( f\right) }{\left( {s}^{\sharp }\right) }$ as any ${y}_{n}$ such that ${y}_{n+1}\leqslant {y}_{n}$.
    \item[Transfer functions] An abstract semantics ${i}^{\sharp }$ of the instruction $i\in \mathcal{I}$
  \end{description}
  \par
  The abstract semantics ${p}^{\sharp }$ of the program $p$ is defined by induction on $p$ where the base case is given by the transfer functions. The correction of ${p}^{\sharp }$ follows from the properties stated above.
  \begin{itemize}
    \item ${\left( a;b\right) }^{\sharp }{\left( {s}^{\sharp }\right) }={b}^{\sharp }{\left( {a}^{\sharp }{\left( {s}^{\sharp }\right) }\right) }$
    \item ${\left( a+b\right) }^{\sharp }{\left( {s}^{\sharp }\right) }=\left( {a}^{\sharp }{\left( {s}^{\sharp }\right) }\right) \sqcup \left( {b}^{\sharp }{\left( {s}^{\sharp }\right) }\right) $
    \item ${1}^{\sharp }{\left( {s}^{\sharp }\right) }={s}^{\sharp }$
    \item ${0}^{\sharp }{\left( {s}^{\sharp }\right) }$ can be chosen arbitrarily
    \item ${\left( {a}^{*}\right) }^{\sharp }{\left( {s}^{\sharp }\right) }=\mathsf{pfp}{\left( {a}^{\sharp }\right) }{\left( {s}^{\sharp }\right) }$
  \end{itemize}
  \par
  \par
  \begin{example}
    Let us define an abstract domain for the example language of Section~\ref{latex_lib_label_1}: we shall abstract the state -- a single integer -- by the signs it may take. More precisely, we take for ${S}^{\sharp }$ the non-empty sets in $\mathcal{P}{\left( \left\{ -,0,+\right\} \right) }$ and the concretisation is defined as:
    \begin{itemize}
      \item $\gamma {\left( {s}^{\sharp }\right) }=\left\{ n\,{\in }\,\mathbb{Z}\mid \mathsf{sign}{\left( n\right) }\in {s}^{\sharp }\right\} $
    \end{itemize}
    The abstract transfer function for guard instructions constrain the abstract state to the relevant signs.
    \begin{itemize}
      \item ${\mathsf{guard}}^{\sharp }{\left( \mathsf{pos}\right) }{\left( {s}^{\sharp }\right) }={s}^{\sharp }\cap \left\{ +\right\} $
      \item ${\mathsf{guard}}^{\sharp }{\left( \mathsf{npos}\right) }{\left( {s}^{\sharp }\right) }={s}^{\sharp }\cap \left\{ -,0\right\} $
    \end{itemize}
    The abstract transfer function for the decrementing command maps positive to non-negative and non-positive to negative:
    \begin{itemize}
      \item ${\mathsf{dec}}_{0}{\left( +\right) }=\left\{ 0,+\right\} $
      \item ${\mathsf{dec}}_{0}{\left( 0\right) }=\left\{ -\right\} $
      \item ${\mathsf{dec}}_{0}{\left( -\right) }=\left\{ -\right\} $
      \item ${\mathsf{dec}}^{\sharp }{\left( {s}^{\sharp }\right) }=\bigcup _{\mbox{\scriptsize{$x$${\in}$${s}^{\sharp }$}}}{\mathsf{dec}}_{0}{\left( x\right) }$
    \end{itemize}
    Since the abstract state space is a powerset, we can use union as the abstract join, and since it is finite, union is also a widening:
    \begin{itemize}
      \item ${s}^{\sharp }\nabla {t}^{\sharp }={s}^{\sharp }\sqcup {t}^{\sharp }={s}^{\sharp }\cup {t}^{\sharp }$
    \end{itemize}
    \par
    Now that the abstract domain is set up, let us run the abstract interpretation on the program $d$ from Section~\ref{latex_lib_label_1} with the input state $\left\{ 0,+\right\} $:
    \begin{enumerate}
      \item Entering the loop with initial state $\left\{ 0,+\right\} $
      \item Applying ${\mathsf{guard}}^{\sharp }{\left( \mathsf{pos}\right) }$: state becomes $\left\{ +\right\} $
      \item Applying ${\mathsf{dec}}^{\sharp }$: state becomes $\left\{ 0,+\right\} $
      \item Invariant found after one iteration: $\left\{ 0,+\right\} \cup \left\{ 0,+\right\} =\left\{ 0,+\right\} $
      \item Applying ${\mathsf{guard}}^{\sharp }{\left( \mathsf{npos}\right) }$: final state is $\left\{ 0\right\} $
    \end{enumerate}
  \end{example}
  \par
  \section{Data refinement}
  \par
  Refinement calculus~\cite{Back1998,VonWright1994} is a discipline to prove the correctness of imperative programs, in a spirit close to Hoare logic. It arises from the remark that, if most predicate transformers do not represent programs, they still represent program \emph{specifications}. Specifications are then \emph{refined} into more precise specifications, and eventually into programs.
  \par
  A key point of the refinement calculus is that the refined specification need not act on the same state as the abstract one. It is typical to use ideal objects -- like multisets -- on the abstract side, and more concrete datatypes -- like linked lists -- on the refined side.
  \par
  We say~\cite{VonWright1994} that $a\in \mathcal{P}{\left( A\right) }{\rightarrow }^{+}\mathcal{P}{\left( A\right) }$ is refined by $b\in \mathcal{P}{\left( B\right) }{\rightarrow }^{+}\mathcal{P}{\left( B\right) }$ through the \emph{coupling invariant} $\iota \in \mathcal{P}{\left( A\right) }{\rightarrow }^{+}\mathcal{P}{\left( B\right) }$, written $a\sqsubseteq _{\iota }b$, when $\iota ;a\sqsubseteq b;\iota $. Intuitively $\iota $ is an action which transforms concrete states into abstract states, so $\iota ;a\sqsubseteq b;\iota $ reads ``doing $b$ then abstracting the state is more precise than abstracting the state then doing $a$''. To emphasise that the type of the state has changed, this relation is often called a \emph{data refinement}.
  \par
  \par
  \begin{example}
    Specifications of imperative programs are typically given as pairs of a precondition and a postcondition. For instance: under the precondition that the initial state is a non-positive integer, the postcondition that the state is $0$ holds after the program has been run. Both preconditions and postconditions can be expressed systematically as (backward) predicate transformers; they can be paired up into a full specification using sequence:
    \begin{itemize}
      \item ${F}_{\mathsf{post}}{\left( X\right) }=\left\{ p\,{\in }\,\mathbb{Z}\mid 0\in X\right\} $
      \item ${F}_{\mathsf{pre}}{\left( X\right) }=\left\{ n\,{\in }\,\mathbb{Z}\mid n\leqslant 0\land n\in X\right\} $
      \item $F={F}_{\mathsf{pre}};{F}_{\mathsf{post}}$
    \end{itemize}
    So that ${F}_{\mathsf{post}}{\left( X\right) }$ is either all of $\mathbb{Z}$ if $0\in X$ or the empty set otherwise, and ${F}_{\mathsf{pre}}{\left( X\right) }$ simply ignores the states in $X$ which do not verify the precondition.
    \par
    The program $d$ from Section~\ref{latex_lib_label_1} meets the specification $F$, however, the state is represented as the \emph{opposite} integer. Hence we have an $\iota $ which reflects this representation:
    \begin{itemize}
      \item ${\iota }_{0}{\left( n\right) }=-n$
      \item $\iota ={{\iota }_{0}}_{*}$
    \end{itemize}
    As per the definition of refinement, the statement that the program $d$ implements the specification reads
    \begin{itemize}
      \item $F\sqsubseteq _{\iota }\mathsf{while}~\mathsf{pos}~\mathsf{do}~\mathsf{decr}$
    \end{itemize}
    It is equivalent to the statement that the precondition entails the weakest liberal precondition of $d=\mathsf{while}~\mathsf{pos}~\mathsf{do}~\mathsf{decr}$:
    \begin{itemize}
      \item $\forall {n}^{\in \mathbb{Z}}.\,\,n\geqslant 0\implies n\in \mathsf{wlp}{\left( d\right) }{\left( \left\{ 0\right\} \right) }$
    \end{itemize}
    which is the typical proof obligation in a Hoare logic setting.
    \par
    The take away from data refinement is that it does not matter what coupling invariant is used, as long as \emph{all the function use the same coupling invariant}. Or, more realistically, under some separation property, if all the function \emph{which have access to some part $A$ of the state} all have coupling invariants which agree on $A$.
    \par
    In practice there are two reasons to refine the type of (a part of) the state: it may be that it is an ideal type, say finite sets of integer, which may be refined into an actual concrete data type, for instance list of integers. Or it may be that the proposed data type is not efficient, and will be refined into a more efficient representation -- list of integers could be refined into binary trees.
  \end{example}
  \par
  \par
  \section{Abstract interpretation in refinement calculus}
  \par
  The main result of this article is that abstract interpretation can be characterised in the language of the refinement calculus: an abstract interpretation of a program $p$ is a \emph{specification} verified by $p$ which is also a function.
  \par
  \begin{theorem}
    \label{latex_lib_label_4}The soundness condition of abstract interpretation is a refinement condition:
         ${{p}^{\sharp }}^{-1}\sqsubseteq _{\left\langle \gamma \right\rangle }p\iff \forall {{s}^{\sharp }}^{\in {\mathcal{S}}^{\sharp }}.\,\,\gamma {\left( {s}^{\sharp }\right) }\subseteq p{\left( \gamma {\left( {p}^{\sharp }{\left( {s}^{\sharp }\right) }\right) }\right) }$
  \end{theorem}
  \begin{proof}
    We have the following equivalent characterisation, thanks to the Galois connection properties:
    \begin{itemize}
      \item ${{p}^{\sharp }}^{-1}\sqsubseteq _{\left\langle \gamma \right\rangle }p\iff {{p}^{\sharp }}^{-1};\left[ \gamma \right] \sqsubseteq \left[ \gamma \right] ;p$
    \end{itemize}
    From which it follows that:
    \begin{list}{}{}
      \item ${{p}^{\sharp }}^{-1}\sqsubseteq _{\left\langle \gamma \right\rangle }p$
      \item ${\Longleftrightarrow}${\small{${\quad}$(Definition of sequence)}}
      \item $\forall {Y}^{\in \mathcal{P}{\left( \mathcal{S}\right) }}.\,\,{{p}^{\sharp }}^{-1}{\left( \left[ \gamma \right] {\left( Y\right) }\right) }\subseteq \left[ \gamma \right] {\left( p{\left( Y\right) }\right) }$
      \item ${\Longleftrightarrow}${\small{${\quad}$(Definition of inclusion)}}
      \item $\forall {Y}^{\in \mathcal{P}{\left( \mathcal{S}\right) }},{{s}^{\sharp }}^{\in {\mathcal{S}}^{\sharp }}.\,\,{s}^{\sharp }\in {{p}^{\sharp }}^{-1}{\left( \left[ \gamma \right] {\left( Y\right) }\right) }\implies {s}^{\sharp }\in \left[ \gamma \right] {\left( p{\left( Y\right) }\right) }$
      \item ${\Longleftrightarrow}${\small{${\quad}$(Definition of $\left[ \gamma \right] $ and basic property of ${{p}^{\sharp }}^{-1}$)}}
      \item $\forall {Y}^{\in \mathcal{P}{\left( \mathcal{S}\right) }},{{s}^{\sharp }}^{\in {\mathcal{S}}^{\sharp }}.\,\,{p}^{\sharp }{\left( {s}^{\sharp }\right) }\in \left[ \gamma \right] {\left( Y\right) }\implies \gamma {\left( {s}^{\sharp }\right) }\subseteq p{\left( Y\right) }$
      \item ${\Longleftrightarrow}${\small{${\quad}$(Definition of $\left[ \gamma \right] $)}}
      \item $\forall {Y}^{\in \mathcal{P}{\left( \mathcal{S}\right) }},{{s}^{\sharp }}^{\in {\mathcal{S}}^{\sharp }}.\,\,\gamma {\left( {p}^{\sharp }{\left( {s}^{\sharp }\right) }\right) }\subseteq Y\implies \gamma {\left( {s}^{\sharp }\right) }\subseteq p{\left( Y\right) }$
      \item ${\Longleftrightarrow}${\small{${\quad}$(${\Rightarrow}$ by $Y=\gamma {\left( {p}^{\sharp }{\left( {s}^{\sharp }\right) }\right) }$ and ${\Leftarrow}$ by monotonicity of $p$)}}
      \item $\forall {{s}^{\sharp }}^{\in {\mathcal{S}}^{\sharp }}.\,\,\gamma {\left( {s}^{\sharp }\right) }\subseteq p{\left( \gamma {\left( {p}^{\sharp }{\left( {s}^{\sharp }\right) }\right) }\right) }$
    \end{list}
  \end{proof}
  \par
  In~\cite{Cousot1992a}, Cousot \& Cousot describe abstract interpretation of inference rule systems. Their approach to defining abstract interpretation resembles refinement calculus, they use, in particular, the remark that inference rule systems can be represented as predicate transformers. Theorem~\ref{latex_lib_label_4} further illuminates the connection.
  \par
  Although so far we have mostly considered forward analyses, a similar characterisation to Theorem~\ref{latex_lib_label_4} holds for backward analysis:
  \begin{theorem}
    ${p}_{*}^{\sharp }\sqsubseteq _{\left\langle \gamma \right\rangle }p\iff \forall {{s}^{\sharp }}^{\in {\mathcal{S}}^{\sharp }}.\,\,\gamma {\left( {p}^{\sharp }{\left( {s}^{\sharp }\right) }\right) }\subseteq p{\left( \gamma {\left( {s}^{\sharp }\right) }\right) }$
  \end{theorem}
  \par
  In traditional refinement calculus, the process consists in starting with an abstract definition, and refine it towards a more concrete definition, weakening the preconditions, strengthening the postconditions while making the state more suitable for execution. In static analysis, refinement calculus is used somewhat backwards: starting from a concrete implementation, it is refined into a more abstract definition, in effect strengthening the precondition and weakening the postconditions, while still making the state more suitable for execution.
  \par
  \section{Conclusion}
  A previous work by Sylvain Boulmé and Michaël Périn~\cite{Boulme2013} uses refinement calculus as a mean to check, in Coq, the correctness of a certificate validation procedure for certificate meant to be output by an abstract interpreter. Although this work is at the intersection of abstract interpretation and refinement calculus, it does not try to establish a connection between refinement calculus and the correctness condition of the abstract interpretation procedure.
  \par
  The present article shows that the language of abstract interpretation can be recast in terms of the refinement calculus. This has been used in the formalisation of Cosa~\cite{cosa}, a Coq verified implementation of an abstract domain for shape analysis. Cosa targets Compcert C~\cite{compcert}, and uses numerical domains by David Pichardie \& \emph{al}~\cite{Blazy2013}.
  \par
  Cosa relies on a variant of the refinement calculus introduced by Peter Hancock based not on predicate transformers but on so-called \emph{interaction structures}~\cite{Hancock}. Compared to predicate transformers, interaction structures carry more information: the set of predicate transformers can be seen as a quotient of the set of interaction structures. The additional information contained in interaction structures can be used to derive a datatype of \emph{strategies} which the oracle is charged with providing, hence formalising the separation between the oracle, which has no bearing on the correctness and does not need to be mechanically verified, and the rules constituting the domain which ensure correctness.
  \par
  Interaction structures were initially developed as a variant of refinement calculus suitable for type theory. Thanks to the results of this article, interaction structures can be also leveraged for abstract interpretation.
  \bibliography{library}
\end{document}